\input ./style/arxiv-vmsta.cfg
\documentclass[numbers,compress,v1.0.1]{vmsta}
\usepackage{vtexbibtags}

\volume{5}
\issue{2}
\pubyear{2018}
\firstpage{145}
\lastpage{165}
\aid{VMSTA100}
\doi{10.15559/18-VMSTA100}
\articletype{research-article}

\ifdefined\HCode 
\def\newmathds#1{\mathbb{#1}}
\else 
\def\newmathds#1{\mathds{#1}}
\fi

\startlocaldefs
\newcommand{\rrVert}{\Vert}
\newcommand{\llVert}{\Vert}
\newcommand{\rrvert}{\vert}
\newcommand{\llvert}{\vert}
\allowdisplaybreaks
\endlocaldefs

\usepackage{dsfont,mathbh}
\usepackage{authorquery}

\newtheorem{theorem}{Theorem}[section]
\newtheorem{lemma}[theorem]{Lemma}
\newtheorem{proposition}[theorem]{Proposition}

\newtheorem{definition}[theorem]{Definition}
\newtheorem{remark}[theorem]{Remark}

\begin{aqf}
\querytext{Q1}{Check if square integrability is really needed for a
continuously differentiable function.}
\querytext{Q2}{Something seems to be missing here.}
\end{aqf}
\begin{document}

\begin{frontmatter}
\pretitle{Research Article}

\title{Computation of option greeks under hybrid stochastic volatility
models via Malliavin calculus}

\author{\inits{B.}\fnms{Bilgi}~\snm{Yilmaz}\ead[label=e1]{ybilgi@metu.edu.tr}}

\address{METU, Institute of Applied Mathematics,
\institution{Middle East Technical University}, 6800~Ankara, \cny{Turkey}}

%



\markboth{B. Yilmaz}{Computation of option greeks under hybrid
stochastic volatility models via Malliavin calculus}

\begin{abstract}
This study introduces computation of option sensitivities (Greeks) using
the Malliavin calculus under the assumption that the underlying asset and
interest rate both evolve from a stochastic volatility model and a
stochastic interest rate model, respectively. Therefore, it integrates the
recent developments in the Malliavin calculus for the computation of Greeks:
Delta, Vega, and Rho and it extends the method slightly. The main results
show that Malliavin calculus allows a running Monte Carlo (MC)
algorithm to
present numerical implementations and to illustrate its effectiveness. The
main advantage of this method is that once the algorithms are
constructed, they
can be used for numerous types of option, even if their payoff functions
are not differentiable.
\end{abstract}
%

\received{\sday{2} \smonth{5} \syear{2017}}
\revised{\sday{26} \smonth{2} \syear{2018}}
\accepted{\sday{30} \smonth{3} \syear{2018}}
\publishedonline{\sday{24} \smonth{4} \syear{2018}}
\end{frontmatter}

\section{Introduction}
\label{intro}

In finance, pricing an option is the main issue of managing a trade. However,
once an option is settled, its price does not remain constant. Instead, it
follows a dynamic path during its survival time. Therefore, the market
participants should protect themselves against the unexpected price changes
by managing the variations in the option price.

The price risk and its management is always inseparably associated with
the Greeks, which are derivatives of an option price with respect to
its certain underlying parameters. The information that the Greeks
contain is used to measure the unexpected option price changes based on
specific risk factors. From this point of view, the Greeks are commonly
used to construct a replicating portfolio to protect the main portfolio
against some of the possible changes related to the certain risk
factors. Thus, the computation of Greeks is more important than
obtaining the price of an option~\cite{glasserman2013}, and it becomes
a fundamental research area in mathematical finance. Here, it is worth
to emphasize that the change in the underlying asset price is a very
significant risk factor that affects the option price virtually. Hence,
among other Greeks, Delta, which measures the change in the option
price for a unit change in the price of underlying asset price, is an
essential indicator to determine the balance of underlying asset and
options which is a hedging ratio.

Over the past two decades, there has been an increasing attention
focused on the accurate pricing of hybrid products that are based on a
combination of underlying assets from different classes~\cite
{Grzela2012}. In light of the recent developments in the hybrid
products, it has become difficult to ignore the computation of Greeks
from which the underlying asset evolves. Thus, the study considers an
extension of stochastic volatility models, namely Hybrid Stochastic
Volatility (HSV) models. Within these models, the interest rate evolves
from a stochastic process rather than from a constant or a
deterministic function to have a more flexible multi-factor stochastic
volatility (SV) model. However, the computation of Greeks under these
assumptions becomes complicated. It is because these suggested models
do not have an explicit distribution or because in most cases the
option payoffs are not differentiable. Therefore, the likelihood method
and the pathwise method are not appropriate ones in computation of
Greeks for this kind of underlying assets. On the other hand, one may
use the finite difference method since it relies on the approximation
of a derivative as the change in a dependent variable over a small
interval of the independent variable, and it can be written using a
small set of difference operators. However, it is computationally
expensive and problematic for discontinuous option payoff functions. In
this respect, the computation of Greeks in hybrid products is not
straightforward as in the Black--Scholes model (BSM).

After the pioneering studies~\cite{Fournie1999,Fournie2001}, the
integration by parts formula in the context of the Malliavin calculus
has come to be considered as one of the main tools in the computation
of Greeks. Since then, a considerable amount of literature has been
published in this research area including some numerical applications.
In most of these studies, the market dynamics are assumed to follow the
BSM assumptions. However, recently, there has been an increasing
interest in the computation of Greeks under the assumptions of SV
models~\cite{Benhamou2002,Ewald2006,Mhlanga} and jump-diffusion type
models~\cite{morel2006,Davis2006,privault2004}.

The theory of the Malliavin calculus is attractive among both
theoreticians and practitioners for three main reasons. First, it
allows researchers to derive explicit weights to design an efficient MC
algorithm. Second, it requires neither any differentiability condition
on payoff functions (unlike the pathwise method) nor probability
density functions of underlying assets (unlike the likelihood method).
Third, it is more flexible than other methods in the sense that
different perturbations of the Brownian motion yield different weights,
which leads to an a priori interest in picking weights with small
variances. Hence, it may be used for all types of option, whether they
have a continuous or discontinuous payoff function.

The driving goal of this study is to introduce an expression for the
computation of Greeks via Malliavin Calculus under the HSV model
dynamics, and then, to use an MC algorithm for the expected values and
their differentials that correspond to the Greeks. The study shows
their accuracy by comparing the findings with the findings of finite
difference method in all variations. The novel part of this study is to
provide applicable formulas using the technical and theoretical results
in the context of the Malliavin calculus. To derive general formulas
for the Greeks, first, an HSV model is defined in such a way that it
allows researchers to obtain generalized Greeks for all types of hybrid
models using fundamental tools in the Malliavin calculus, namely
integration by parts and the Bismut--Elworthy--Li formulas on the
Gaussian space. The main results reveal that the Malliavin calculus
gives the Greeks as a product of an option payoff function and an
independent weight function called the Malliavin weight. Moreover, the
numerical illustrations reveal that the computation cost for the
Malliavin calculus method is less than the computation cost for the
finite difference method in all variations.

This article consists of 4 sections. In Section~\ref{sec:MCgreeks}, it
introduces a general formula for the HSV models and the perturbed price
processes. Then, it presents the main results which extend the existing
theoretical Greeks formulas. In Section~\ref{sec:num}, the formulas are
derived for the Heston stochastic volatility model with a stochastic
interest rate, namely, the Vasicek model. Also in this section,
numerical illustrations are presented for Delta, Rho, and Vega of a
European call option to verify the efficiency and to compare findings
with the outcomes of the finite difference method in all variations.
Finally, Section~\ref{sec:conclusion} is the conclusion.
\section{Computation of the generalized greeks}
\label{sec:MCgreeks}
Consider a fixed filtered probability space $ (\varOmega, \mathcal
{F},\mathcal{F}_{t\in [0,T ]},\mathbb{Q} )$ which is rich
enough to accommodate a Brownian motion of dimension three for the
computation convenience. Note that in this representation $\mathcal
{F}_{t}$ is a filtration generated by three independent standard
Brownian motions $W_{t}^i$ for $i=1,2,3$ and $\mathbb{Q}$ is the
risk-neutral probability measure.

Throughout the study, it is assumed that the underlying asset price
evolves from the SDE system
\begin{align}
dS_t&=r_{t}S_{t}\;dt+S_{t} \,
\sigma (V_t )dZ_{t}^{1} \label
{eq:general1},
\\
dV_t&=u (V_t )dt+v (V_t )
dZ_{t}^2 \label
{eq:general2},
\\
dr_t&=f (r_t )dt+g (r_t )
dZ_{t}^3 \label{eq:general3},
\end{align}
where $(Z_{t})_{t\in [0, T ]}^i$'s are correlated Brownian
motions with correlation coefficients $\rho_{ij} \in(-1, 1)$ for
$i,j=1,2,3$. These correlated Brownian motions may also be represented
by a combination of three independent Brownian motions $(W_{t})_{t\in
 [0, T ]}^i$ such as
\begin{align*}
dZ_{t}^1 &= dW_{t}^1,
\\
dZ_{t}^2 &= \rho_{12}\;dW_{t}^{1}+
\mu_1 \;dW_{t}^{2},
\\
dZ_{t}^3 &= \rho_{13} dW_{t}^1
+ \mu_2 dW_{t}^2 + \mu_3
dW_{t}^{3},
\end{align*}
where the parameters are\vadjust{\goodbreak}
\begin{align*}
\mu_1&=\sqrt{1-\rho_{12}^2}, \qquad
\mu_2= \frac{\rho_{23}-\rho_{12} \rho
_{13}}{\mu_1},\\
\mu_3&= \frac{\sqrt{1-\rho_{12}^2 - \rho_{13}^2 -
\rho_{23}^2 + 2\rho_{13}\rho_{12} \rho_{23}}}{ \mu_1}.
\end{align*}
Throughout the study, it is assumed that the correlation coefficients
$\rho_{ij}$ are chosen in such a way that $\mu_3$ is a real number. The
solutions $S_t$, $V_t$, and $r_t$ represent an underlying asset,
volatility and interest rate processes with initial values $S_0, V_0$,
and $r_0$, respectively. Here, it is assumed that $\sigma, u,v, f$, and
$g$ are continuously differentiable functions with bounded derivatives
of order at least two. Moreover, $\sigma$, $v$, and $g$ are assumed to
be adapted and not equal to zero everywhere in the domain.

Now, suppose the SDEs given in equations (\ref{eq:general1}), (\ref
{eq:general2}), and (\ref{eq:general3}) can be merged into a
three-dimensional SDE system and it is represented as
\begin{align}
\label{eq:multiSDE} dX_t&=\beta (X_t )dt+a (X_t
)d\mathbb{W}_t,
\\
X_0&=x,
\nonumber
\end{align}
where,
\begin{eqnarray*}
X_t= %
\begin{bmatrix}
S_t \\[0.3em]
V_t \\[0.3em]
r_t
\end{bmatrix} %
,\quad \beta
(X_t )= %
\begin{bmatrix}
r_tS_t \\[0.3em]
u (V_t ) \\[0.3em]
f (r_t )
\end{bmatrix} %
,
\end{eqnarray*}
and
\begin{align*}
a (X_t )&= %
\begin{bmatrix}
S_t\sigma (V_t ) & 0 & 0 \\[0.3em]
\rho_{12} v (V_t ) & \mu_1 v (V_t ) & 0 \\[0.3em]
\rho_{13} g (r_t ) & \mu_2 \; g (r_t ) & \mu_3 g
(r_t )
\end{bmatrix} %
,\\
x&= %
\begin{bmatrix}
S_0 \\[0.3em]
V_0 \\[0.3em]
r_0
\end{bmatrix} %
,\qquad \mathbb{W}_t=
\begin{bmatrix}
W_t^1 \\[0.3em]
W_t^2 \\[0.3em]
W_t^3
\end{bmatrix} %
.
\end{align*}
Notice that $X_t$ is a Markov process and $(\mathbb{W}_t)_{t\in[0, T]}$
is a three-dimensional standard Brownian motion.

In this representation, it is convenient to assume that $\beta$ and $a$
both are at least twice continuously differentiable functions with
bounded derivatives and adapted for the sake of computations. Moreover,
to ensure the existence of a unique strong solution to equation (\ref
{eq:multiSDE}), it is assumed that both $\beta$ and $a$ satisfy the
Lipschitz and polynomial growth conditions. Under these assumptions,
the study also guarantees that $X_t$ is a Markov process, the
trajectories of this solution are almost surely continuously
differentiable for all $t$ up to explosion time~\cite{protter2005}.

Furthermore, in order to have an appropriate solution, it should also
be guaranteed that the diffusion matrix $a$ satisfies the uniform
ellipticity condition~\cite{Mhlanga},
\begin{equation}
\label{eqn:ellipticity} \bigl(a(x)\zeta\bigr)^\top\bigl(a(x)\zeta\bigr)\geq
\epsilon \llVert \zeta \rrVert ^2,\quad \text{for any}\; x,\zeta\in
\mathbb{R}^3.
\end{equation}

Here, it is worth to emphasize that if the conditions on $\mu_{i,j}$
and the
functions $\sigma, v, g$ are satisfied, the diffusion matrix $a$
satisfies the
uniform ellipticity condition.\vadjust{\goodbreak} This implies that $a$ is a positive definite
matrix, and its eigenvalues are greater than a small positive integer
$\epsilon\in\mathbb{R}$. Hence, $a$ is an invertible matrix, and its inverse
is bounded. This condition also implies that for any bounded function
$\gamma:[0,T]\times\mathbb{R}^3\mapsto\mathbb{R}^3$, $a^{-1}\gamma$ is
bounded, and further, $(a^{-1}\gamma)(X_t)$ lies in the Hilbert space
$L^2([0,T]\times\varOmega)$~\cite{Fournie1999}.

The first variation process of $(X_t)_{t\in[0, T]}$, which is the
derivative of $(X_t)_{t\in[0, T]}$ with respect to its initial value
$x$, plays an important role in the computation of Greeks. Hence, the
study reminds the definition of the first variation process.

\begin{definition}[First variation process]
\label{def:firstvar}
Let $X_t$ be a process given by equation (\ref{eq:multiSDE}). Then, the
first variation of this process is defined by
\begin{align}
\label{eq:firstvariation} dY_t&= \beta^\prime (X_t
)Y_tdt+\sum_{i=1}^{3}{a_{i}^\prime
(X_t )Y_t \, d W_t^{i}},
\\
Y_0&=\mathbh{1}_{3\times3}
\nonumber
,
\end{align}
where $\beta^\prime$ and $a_i^\prime$ are the Jacobian of $\beta$ and
$i$th column vector of matrix $a$ with respect to $x$, respectively.
Here, $\mathbh{1}_{3\times3}$ is the identity matrix of $\mathbb{R}^3$,
and $Y_t=D^x X_t$. Note that $\beta$ and $a$ are assumed to be at least
twice continuously differentiable functions with bounded derivatives.
Moreover, the diffusion matrix $a$ satisfies the uniform ellipticity
condition (\ref{eqn:ellipticity}), and $X_t$ has continuous trajectories.
\end{definition}
Now, it is the time to introduce the following lemma for the integrity
of the study.
\begin{lemma}
If $(Y_t Y_s^{-1} a) \in L^2([0,T]\times\varOmega)$ for all $s,t \in
[0,T]$, then $X_t$ is Malliavin differentiable and the Malliavin
derivative of $X_t$ can be written as follows
\begin{equation}
\label{eq:malderX} D_sX_t=Y_tY_s^{-1}a
(X_s )\mathbh{1}_{s\leq t}, \quad s\geq0, \; a.s.
\end{equation}
\end{lemma}
One can calculate the components of the first variation process as in
the following proposition.
\begin{proposition}
\label{prop:firstvariation}
Let $X_t$ and the first variation process $Y_t$ for $t \in[0,T]$ be
defined by equations (\ref{eq:multiSDE}) and (\ref{eq:firstvariation}),
respectively. Then, $Y_{t}^{ij} = 0 $ and $Y_t^{23}=0$ a.s. if $i > j $
for $i,j=1,2,3$. Moreover, $Y_t^{11}, Y_{t}^{22}$, and $Y_{t}^{33}$
have the following solutions for $t \in[0,T]$
\begin{align}
Y_{t}^{11}&= \exp \Biggl( \int_0^t
{ \biggl(r_s -\frac{1}{2}\sigma^2
(V_s ) \biggr)} \; ds + \int_0^t {
\sigma (V_s ) } \, dZ_{s}^1 \Biggr),
\label{eq:Yt11}
\\
Y_{t}^{22} &= \exp \Biggl( \int_0^t
{ \biggl(u^{\prime} (V_s ) -\frac{1}{2}v^{\prime}
(V_s )^2 \biggr)} \; ds + \int_0^t
{ v^{\prime
} (V_s )} dZ_{s}^2 \Biggr),
\\
Y_{t}^{33} &= \exp \Biggl( \int_0^t
{ \biggl(f^{\prime} (r_s ) -\frac{1}{2}g^{\prime}
(r_s )^2 \biggr)} \; ds + \int_0^t
{ g^{\prime
} (r_s )} dZ_{s}^3 \Biggr).
\end{align}
Furthermore, $Y_t^{12}$ and $Y_{t}^{13}$ satisfy
\begin{align*}
dY_t^{12}&= r_t Y_t^{12}
dt + \bigl[ \sigma (V_t ) Y_t^{12} +
S_t \, \sigma^{\prime} (V_t ) Y_t^{22}
\bigr] dW_t^{1},
\\
dY_t^{13}&= \bigl[ r_t Y_t^{13}
+ S_t Y_t^{33} \bigr] dt + \sigma
(V_t ) Y_t^{13} dW_t^{1},
\end{align*}
with the initial values $Y_0^{12}=Y_0^{13}=0$.
\end{proposition}
\begin{proof}
In order to find the entries of the system of SDEs which is satisfied
by $Y$, one has first to calculate the Jacobian of $\beta$ and the
$i$th column vector, $a_i$, of the diffusion matrix $a$ for
$i=1,2,3$. These are
\begin{align*}
\beta^\prime (x_1, x_2, x_3 )&=
\begin{bmatrix}
x_3 & 0 & x_1 \\[0.3em]
0 & u^{\prime}(x_2) & 0 \\[0.3em]
0 & 0 & f^{\prime}  (x_3 )
\end{bmatrix} %
,
\\
a_{1}^\prime (x_1, x_2,
x_3 )&= %
\begin{bmatrix}
\sigma (x_2 ) & x_1 \sigma^{\prime}(x_2) & 0 \\[0.3em]
0 & \rho_{12} v^{\prime}(x_2) & 0 \\[0.3em]
0 & 0 & \rho_{13} g^{\prime}(x_3)
\end{bmatrix} %
,\\
a_{2}^\prime (x_1, x_2,
x_3 )&= %
\begin{bmatrix}
0 & 0 & 0 \\[0.3em]
0 & \mu_1 v^{\prime}(x_2) & 0 \\[0.3em]
0 & 0 & \mu_2 g^{\prime} (x_3)
\end{bmatrix} %
,
\\
a_{3}^\prime (x_1, x_2,
x_3 )&= %
\begin{bmatrix}
0 & 0 & 0 \\[0.3em]
0 & 0 & 0 \\[0.3em]
0 & 0 & \mu_3 g^{\prime} (x_3 )
\end{bmatrix} %
.
\end{align*}
By inserting these equations into equation (\ref{eq:firstvariation}),
one obtains the first variation process as
\begin{align*}
& %
\begin{bmatrix}
dY_{t}^{11} & dY_{t}^{12} & dY_{t}^{13} \\[0.3em]
dY_{t}^{21} & dY_{t}^{22} & dY_{t}^{23} \\[0.3em]
dY_{t}^{31} & dY_{t}^{32} & dY_{t}^{33}
\end{bmatrix}\\ %
&\quad= %
\begin{bmatrix}
r_t Y_{t}^{11} + S_t Y_{t}^{31} & r_t Y_{t}^{12} + S_t Y_{t}^{32} &
r_t Y_{t}^{13} + S_t Y_{t}^{33} \\[0.3em]
u^{\prime} (V_t) Y_{t}^{21} & u^{\prime} (V_t) Y_{t}^{22} & u^{\prime
}(V_t) Y_{t}^{23} \\[0.3em]
f^{\prime} (r_t) Y_{t}^{31} & f^{\prime} (r_t) Y_{t}^{32} & f^{\prime}
(r_t) Y_{t}^{33}
\end{bmatrix} %
dt
\\
&\quad +  \small{ \left[%
\begin{array}{@{}c@{\ }c@{\ }c@{}}
\sigma(V_t) Y_{t}^{11} + S_t \sigma^{\prime} (V_t) Y_{t}^{21} & \sigma
(V_t) Y_{t}^{12} + S_t \sigma^{\prime} (V_t) Y_{t}^{22}
& \sigma(V_t) Y_{t}^{13} + S_t \sigma^{\prime} (V_t) Y_{t}^{23} \\[0.3em]
\rho_{12} v^{\prime} (V_t) Y_{t}^{21} & \rho_{12} v^{\prime} (V_t)
Y_{t}^{22} & \rho_{12} v^{\prime}(V_t) Y_{t}^{23} \\[0.3em]
\rho_{13} g^{\prime} (r_t) Y_{t}^{31} & \rho_{13} g^{\prime} (r_t)
Y_{t}^{32} & \rho_{13} g^{\prime} (r_t) Y_{t}^{33}
\end{array}\right] %
dW_t^1
}
\\
&\quad +  %
\begin{bmatrix}
0 & 0
& 0 \\[0.3em]
\mu_1 v^{\prime} (V_t) Y_{t}^{21} & \mu_{1} v^{\prime} (V_t)
Y_{t}^{22} & \mu_1 v^{\prime}(V_t) Y_{t}^{23} \\[0.3em]
\mu_2 g^{\prime} (r_t) Y_{t}^{31} & \mu_2 g^{\prime} (r_t) Y_{t}^{32}
& \mu_2 g^{\prime} (r_t) Y_{t}^{33}
\end{bmatrix} %
dW_t^2
\\
&\quad +  %
\begin{bmatrix}
0 & 0
& 0 \\[0.3em]
0 & 0 & 0 \\[0.3em]
\mu_3 g^{\prime} (r_t) Y_{t}^{31} & \mu_3 g^{\prime} (r_t) Y_{t}^{32}
& \mu_3 g^{\prime} (r_t) Y_{t}^{33}
\end{bmatrix} %
dW_t^3,
\end{align*}
with an initial value of the identity matrix in $\mathbb{R}^3$.

Then, one can easily deduce that $Y_{t}^{ij} = 0 $ and $Y_t^{23}=0$
a.s. if $i > j $ for $i,j=1,2,3$ by applying the It\^o calculus.
Moreover, for $t \in[0,T]$
\begin{align*}
dY_{t}^{11}&= r_t Y_t^{11}
\; dt + \sigma (V_t ) Y_t^{11}
dZ_{t}^1,
\\
dY_{t}^{22}&= u^{\prime} (V_t )
Y_t^{22}dt + v^{\prime} (V_t )
Y_t^{22} dZ_{t}^2,
\\
dY_{t}^{33}&= f^{\prime} (r_t )
Y_t^{33} dt + g^{\prime} (r_t )
Y_t^{33} dZ_{t}^3,
\end{align*}
with the initial values $Y_0^{11}=Y_0^{22}=Y_0^{33}=1$.
\end{proof}
\begin{remark}
\label{rmrk:delta}
As an immediate conclusion of equation~(\ref{eq:Yt11}), $Y_t^{11}$ may
be written in terms of the asset price, i.e.
\begin{equation*}
Y_t^{11}= \frac{1}{S_0} S_t,\quad a.s.
\end{equation*}
%
\end{remark}
For the computation purposes, it is convenient to consider a continuous
time trading economy with a finite time horizon $t\in[0,T]$. Suppose
that the uncertainty in this economy is idealized by a fixed filtered
probability space $ (\varOmega, \mathcal{F},\mathcal{F}_{t\in
[0,T ]}, \mathbb{Q} )$ and the information evolves according
to the filtration $(\mathcal{F}_t)_{t\in[0,T]}$ generated by $(\mathbb
{W}_t)_{t\in[0, T]}$. Furthermore, consider an option in this market,
which has a square integrable payoff function denoted by $\varPhi=\varPhi
(X_{t_1},\ldots, X_{t_n} )$, where $\varPhi$ is a payoff function
that is continuously differentiable with bounded derivatives. In other
words, the option is replicable since $\newmathds{E}[\varPhi
(X_{t_1},\ldots, X_{t_n} )^2]<\infty$~\cite{etheridge2002}.

Now, following studies~\cite{Davis2006,Ewald2006}, the option price,
$p$, at time $t=0$ with a maturity $T<\infty$ is traditionally
calculated as the expected value of the discounted payoff at maturity
conditionally to the present information, which is described by the
$\sigma$-algebra $\mathcal{F}_0$. Now, it is possible to define the
price process in the following definition.
\begin{definition}
\label{price}
Price, $p(x)$, of an option with an underlying $X$ and a payoff
function $\varPhi$, is defined by
\begin{equation}
\label{eq:valuefunction} p(x)=\newmathds{E} \bigl[e^{-\int_{0}^{T}{r_tdt}}\varPhi (X_{t_1},
\ldots, X_{t_n} )|\mathcal{F}_0 \bigr],
\end{equation}
where, $0=t_1,\ldots,t_n=T$ is a partition of the finite time horizon
$[0, T]$, and $p(x)$ denotes today's fair price of the option. Namely,
in this study $\newmathds{E}[.]$ is the expected value under the
risk-neutral probability.
\end{definition}
Note that the objective of the computation of Greeks is to
differentiate the price, $p$, of an option with respect to model parameters.
\begin{remark}
In the computation of Greeks, it is fundamentally assumed that the
option payoff function, $\varPhi$, is continuously differentiable with
bounded derivatives. However, they are generally not smooth in genuine
markets. In this case, if the law of the underlying asset is absolutely
continuous, and the option payoff function, $\varPhi$, is Lipschitz, it is
possible to derive explicit Malliavin weights for the Greeks by
Proposition~\ref{prp:genchain}.
\end{remark}
In the computation of Greeks via the Malliavin calculus, a weight
function, which is independent of the option payoff function, is
obtained. To obtain a valid computation result, one has to guarantee
that the Malliavin weights do not degenerate with probability one.
Hence,~\cite{Fournie1999} demonstrates the set of square integrable functions
\begin{equation*}
\varGamma_n= \Biggl\{\alpha\in L^2 \bigl( [0, T ]
\bigr);\int_0^{t_i}{\alpha(t)dt}=1, \forall i=1,
\ldots,n \Biggr\},
\end{equation*}
in $\mathbb{R}$, to avoid the degeneracy. Here, $t_i$'s are the given
time values in $ [0, T ]$.

Now, it is the time to remind readers about the celebrated
Bismut--Elworthy--Li formula introduced in~\cite
{bismut1984large,elworthy1994formulae}.
%
\begin{proposition}
\label{prp:origibismut}
Suppose that the functions $\beta$ and $a$ in equation (\ref
{eq:multiSDE}) are continuously differentiable with bounded
derivatives, and the diffusion matrix $a$ satisfies the uniform
ellipticity condition~(\ref{eqn:ellipticity}). Moreover, the option
payoff function, $\varPhi$, is square integrable, and continuously
differentiable with bounded derivatives. Now, consider the price
process given in Definition~\ref{price} with a maturity $T<\infty$, then
\begin{equation}
\label{eqn:bissmut} \bigl(\nabla p (x ) \bigr)^{\top}=\newmathds{E}
\bigl[e^{-\int
_{0}^{T}{r_tdt}}\varPhi (X_{t_1},\ldots,X_{t_n} )\pi
\bigr],
\end{equation}
where $\nabla$, and $\pi$ denote the gradient and Malliavin weight,
respectively. Here,
\begin{equation}
\pi=\int_{0}^{T}{\alpha(t) \bigl(a^{-1}
(X_t )Y_t \bigr)^{\top
} d\mathbb{W}_t},
\end{equation}
where $\alpha\in\varGamma_n$, and $Y_t$ is the first variation process.
\end{proposition}

\subsection{Computation of Delta}
Delta of an option is a measure of variations in its price with respect
to initial underlying asset price, and it determines a hedging ratio.
It can be computed by using Proposition~\ref{prp:Malliavindelta}.
\begin{proposition}
\label{prp:Malliavindelta}
Suppose the functions $\beta$ and $a$ are both as in equation (\ref
{eq:multiSDE}). Moreover, they are continuously differentiable
functions with bounded derivatives and the diffusion matrix $a$
satisfies the uniform ellipticity condition (\ref{eqn:ellipticity}).
Now, consider an option with payoff $\varPhi$, which is a continuously
differentiable function with bounded derivatives. Then, Delta of the
option with the price function (\ref{eq:valuefunction}) is
\begin{equation}
\label{eq:delta123} \varDelta=\newmathds{E} \bigl[\varPhi (X_{t_1},\ldots,
X_{t_n} ) \varDelta _{\mathit{MW}} \bigr],
\end{equation}
where $\varDelta_{\mathit{MW}} $ is the Maliavin weight of Delta and it is
\begin{align*}
\varDelta_{\mathit{MW}}&= \frac{e^{-\int_{0}^{T}{r_tdt}}}{ S_0 T} \Biggl( \int_{0}^{T}{
\frac{1}{ \sigma (V_t )}dW_t^1}- \frac{\rho_{12}} {\mu_1} \int
_{0}^{T}{\frac{1}{\sigma (V_t
)} dW_t^2}
\\
&\quad+ \frac{ \rho_{12} \mu_2 - \rho_{13} \mu_1}{\mu_1 \mu_3
} \int_{0}^{T}{
\frac{1} {\sigma (V_t ) } dW_{t}^3} \Biggr).
\end{align*}
\end{proposition}
\begin{proof}
Using the inverse of the diffusion matrix $a$
\begin{equation*}
a^{-1} (X_t )= %
\begin{bmatrix}
\frac{1}{S_t\sigma (V_t )} & 0 & 0 \\[0.3em]
\frac{-\rho_{12}}{\mu_1 S_t \sigma (V_t )} & \frac{1}{\mu_1
v (V_t )} & 0 \\[0.3em]
\frac{\rho_{12}\mu_2 - \rho_{13} \mu_1} {\mu_1 \mu_3 S_t \sigma(V_t)}
& \frac{-\mu_2}{\mu_1 \mu_3 v (V_t )}& \frac{1}{\mu_3 g
(r_t )}
\end{bmatrix}
,
\end{equation*}
one obtains
\begin{equation*}
 \bigl(a^{-1} (X_t )Y_t
\bigr)^{\top}= %
\begin{bmatrix}
\frac{Y_t^{11}}{S_t\sigma (V_t )} & \frac{-\rho_{12}
Y_t^{11}}{\mu_1 S_t\sigma (V_t )} & \frac{ (\rho_{12} \mu
_2 - \rho_{13} \mu_1 )Y_t^{11}}{ \mu_1 \mu_3 \sigma (V_t
) S_t } \\[0.3em]
\frac{Y_t^{12}}{S_t\sigma (V_t )} & \frac{-\rho
_{12}Y_t^{12}}{\mu_1S_t\sigma(V_t)}+\frac{Y_t^{22}}{\mu_1 v
(V_t )}&\frac{(\rho_{12}\mu_2-\rho_{13}\mu_1)Y_t^{12}}{S_t\sigma
(V_t)\mu_1\mu_3} -\frac{ \mu_2 Y_t^{22}}{ \mu_1 \mu_3 v (V_t
)} \\[0.3em]
\frac{Y_t^{13}}{S_t\sigma (V_t )} & \frac{-\rho_{12}
Y_t^{13}}{\mu_1S_t\sigma(V_t)} & \frac{(\rho_{12}\mu_2-\rho_{13}\mu
_1)Y_t^{13}}{\mu_1\mu_3S_t\sigma(V_t)}+\frac{Y_t^{33}}{ \mu_3 g
(r_t )}
\end{bmatrix} %
.
\end{equation*}
An immediate conclusion of Proposition~\ref{prp:origibismut} with the
special choice of $\alpha(t)=\frac{1}{T}$ is given by 
%
\begin{align*}
\varDelta_{\mathit{MW}}&= \frac{e^{-\int_{0}^{T}{r_tdt}}}{T} \Biggl(\int_{0}^{T}{
\frac{Y_t^{11}}{S_t\sigma (V_t
)}dW_t^1}-\int_{0}^{T}{
\frac{\rho_{12} Y_t^{11}}{ \mu_1 S_t\sigma
(V_t )}dW_t^2}
\\
&\quad+\int_{0}^{T}{ \frac{(\rho_{12} \mu_2 - \rho
_{13} \mu_1) Y_t^{11} } {\mu_1 \mu_3 S_t \sigma (V_t )}
}dW_t^3 \Biggr).
\end{align*}
Here one may choose $\alpha(t)=\frac{1}{T}$ since European options are
priced at maturity and therefore $t_i=T$. Then, with Remark~\ref
{rmrk:delta}, the final result may be obtained easily.
\end{proof}
As it is seen in the above formula, the gradient of the option price is
denoted by $\nabla p(x)=(\frac{\partial p}{\partial S_0}, \frac
{\partial p}{\partial V_0}, \frac{\partial p}{\partial r_0})^{\top}$,
where ${\top}$ denotes the transpose. The first row of the solution
corresponds to the option's Delta. The remaining two rows correspond to
the changes in the price with respect to the initial volatility and
initial interest rate, respectively. Hence, as a consequence of this
result, one can present the following two remarks.
%
\begin{remark}[$\mathit{Vega}^{V_t}$]
The sensitivity of an option to its initial volatility is
\begin{equation*}
\mathit{Vega}^{V_t}= \newmathds{E} \bigl[\varPhi (X_{t_1},\ldots,
X_{t_n} )\mathit{Vega}^{V_t}_{\mathit{MW}} \bigr],
\end{equation*}
where
\begin{align*}
\mathit{Vega}^{V_t}_{\mathit{MW}}&= e^{-\int_{0}^{T}{r_tdt}} \Biggl(\int
_0^T{\frac
{Y_t^{12}}{S_t\sigma (V_t )}dW_t^1}+
\int_0^T{\frac{-\rho
_{12}Y_t^{12}}{\mu_1S_t\sigma(V_t)}+
\frac{Y_t^{22}}{\mu_1 v
(V_t )}dW_t^2}
\\
&\quad+\int_0^T{\frac{(\rho_{12}\mu_2-\rho_{13}\mu
_1)Y_t^{12}}{\mu_1\mu_3\sigma(V_t)}-
\frac{ \mu_2 Y_t^{22}}{ \mu_1 \mu_3
v (V_t )}dW_t^3} \Biggr).
\end{align*}
\end{remark}
%
%
\begin{remark}[$\mathit{Rho}^{r_t}$]
The sensitivity of an option to the initial interest rate is
\begin{equation*}
\mathit{Rho}^{r_t}=\newmathds{E} \bigl[\varPhi (X_{t_1},\ldots,
X_{t_n} )\mathit{Rho}^{r_t}_{\mathit{MW}} \bigr],
\end{equation*}
where
\begin{align*}
\mathit{Rho}^{r_t}_{\mathit{MW}}&=e^{-\int_{0}^{T}{r_t}dt} \Biggl(\int
_0^T{\frac
{Y_t^{13}}{S_t\sigma (V_t )}dW_t^1}-
\int_0^T{\frac{\rho_{12}
Y_t^{13}}{\mu_1S_t\sigma(V_t)}dW_t^2}
\\
&\quad+\int_0^T{\frac{(\rho_{12}\mu_2-\rho_{13}\mu_1)Y_t^{13}}{\mu_1\mu
_3S_t\sigma(V_t)}+
\frac{Y_t^{33}}{ \mu_3 g (r_t )}dW_t^3} \Biggr).
\end{align*}
\end{remark}
%
\subsection{Computation of Rho}
\label{sbs:variationsDrift}
The computation of Rho is not as straightforward as is the computation
of Delta since the interest rate is neither constant nor deterministic.
Hence, instead of directly differentiating the option price with
respect to the interest rate, one may consider adding a perturbation
term $\epsilon$ to the drift term and then try to observe the effect of
the perturbation on the option. Here, it is necessary to clarify what
exactly is meant by a perturbed process $X_t^\epsilon$ to observe the
change in the price with respect to change in the drift term.\vadjust{\goodbreak}

As in~\cite{Fournie1999,Davis2006}, the study introduces the perturbed
process $ (X_t^\epsilon )_{t\in[0, T]}$ as follows:
\begin{equation}
\label{eq:perturb1} dX_t^\epsilon= \bigl(\beta
\bigl(X_t^\epsilon \bigr)+\epsilon\gamma \bigl(X_t^\epsilon
\bigr) \bigr)dt+a \bigl(X_t^\epsilon \bigr)d\mathbb
{W}_t,\qquad X_0^\epsilon=x,
\end{equation}
where $\epsilon$ is a small scalar and $\gamma:[0,T]\times\mathbb
{R}^3\mapsto\mathbb{R}^3$ is a bounded function. Furthermore, $\beta$
and $a$ satisfy the regularity conditions that are discussed above.

To interpret the impact of a structural change in the drift and the
price, one should perturb the price process as in the following definition.
\begin{definition}
Suppose $X_t^\epsilon$ is the solution of the SDE system given in (\ref
{eq:perturb1}) for $t\in [0, T ]$ and $\varPhi$ is a
continuously differentiable function at least order two with bounded derivatives. Then,
the perturbed price process $p^\epsilon (x )$ is given by
\begin{equation}
\label{eq:perturb2} p^\epsilon (x )=\newmathds{E} \bigl[e^{-\int_{0}^{T}{r_t^\epsilon
\;dt}}\varPhi
\bigl(X_{t_1}^\epsilon,\ldots, X_{t_n}^\epsilon
\bigr)|\mathcal{F}_0 \bigr].
\end{equation}
\end{definition}
Now it is convenient to present the following proposition to show the
sensitivity of the option to the parameter $\epsilon$ in the point
$\epsilon=0$.
\begin{proposition}
\label{prp:variationDrift}
Suppose that $\beta$, $a$ are continuously differentiable functions
with bounded derivatives, and $a$ satisfies the uniform ellipticity
condition~(\ref{eqn:ellipticity}). Then, for any square integrable and
continuously differentiable function with bounded derivatives $\varPhi$,
$\epsilon\longmapsto p^\epsilon(x)$ is differentiable at any $x\in
\mathbb{R}^3$ and
\begin{align*}
\frac{\partial p^\epsilon(x)}{\partial\epsilon}|_{\epsilon=0} &= \newmathds{E} \biggl[\varPhi
\bigl(X_{t_1}^\epsilon,\ldots, X_{t_n}^\epsilon
\bigr) \frac{\partial}{\partial\epsilon}e^{-\int_0^T{r_t^\epsilon\;
dt}}|_{\epsilon=0} \biggr]
\\
&\quad+\newmathds{E} \Biggl[e^{-\int_0^T{r_t^\epsilon dt}}\varPhi (X_{t_1},\ldots,
X_{t_n} ) \int_0^T{ \alpha(t)
\bigl(a^{-1}(X_t)\gamma(X_t)
\bigr)^{\top} d\mathbb{W}_t}|_{\epsilon=0} \Biggr].
\end{align*}
\end{proposition}
\begin{proof}
See the proof in~\cite{Fournie1999}.
\end{proof}

\begin{proposition}
\label{prp:Rho}
Suppose $\beta$ and $a$ are continuously differentiable functions with
bounded derivatives. Moreover, $a$ satisfies the uniform ellipticity
condition (\ref{eqn:ellipticity}), and the square integrable option
payoff function $\varPhi$ is a continuously differentiable function with
bounded derivatives. Then, Rho of the option is
\begin{equation}
\label{eq:rho} \mathit{Rho}=\newmathds{E} \bigl[\varPhi (X_{t_1},\ldots,
X_{t_n} ) \mathit{Rho}_{\mathit{MW}} \bigr],
\end{equation}
where
\begin{align*}
\mathit{Rho}_{\mathit{MW}}&= \frac{e^{-\int_{0}^{T}{r_tdt}}}{T} \Biggl(\int_{0}^{T}{
\frac
{dW_t^1}{\sigma (V_t )}}-\frac{\rho_{12}}{\mu_1}\int_{0}^{T}{
\frac{dW_t^2}{\sigma (V_t )}}
\\
&\quad+ \frac{\rho_{12}
\mu_2- \rho_{13} \mu_1}{\mu_1 \mu_3} \int_{0}^{T} {
\frac{dW_t^3}{\sigma
 (V_t )}} - T^2 \Biggr).
\end{align*}
\end{proposition}
\begin{proof}
Rho measures the effect of a change in the interest rate on the option
price. In Proposition~\ref{prp:variationDrift}, there are mainly three
sources of perturbation: drift terms of risky asset, volatility, and
interest rate processes. The function $\gamma(X_t)$ can be chosen as
any combination of these three sources. Since the study is
investigating the effect of the interest rate on the option price, it
should perturb the original drift with\vadjust{\goodbreak} $\gamma (X_t )=
(S_t, 0,0 )^{\top}$. In this case, $\gamma$ is a bounded function
since $t\in[0, T]$ is a continuous time trading economy with a finite
time horizon. Note the fact that the dynamics are given under a
risk-neutral probability measure and perturbation. Hence, the discount
process becomes $e^{-\int_0^T{ (r_t+\epsilon ) dt}}$.

First one should find
\begin{equation*}
\bigl(a^{-1} (X_t )\gamma (X_t )
\bigr)^{\top}= \biggl(\frac{1}{\sigma (V_t )}, -\frac{\rho_{12}}{\mu_1 \sigma
(V_t )},
\frac{\rho_{12} \mu_2 - \rho_{13} \mu_1}{\mu_1 \mu_3
\sigma (V_t )} \biggr).
\end{equation*}
Then, by inserting the equation above into the expectation term, one obtains
\begin{align*}
\mathit{Rho}&=\newmathds{E} \Biggl[e^{-\int_{0}^{T}{r_tdt}}\varPhi (X_{t_1},\ldots,
X_{t_n} ) \frac{1}{T} \Biggl(\int_{0}^{T}{
\frac{dW_t^1}{\sigma
(V_t )}}-\frac{\rho_{12}}{\mu_1}\int_{0}^{T}{
\frac{dW_t^2}{\sigma
 (V_t )}}
\\
&\quad+ \frac{\rho_{12} \mu_2- \rho_{13} \mu_1}{\mu_1 \mu_3} \int_{0}^{T} {
\frac{dW_t^3}{\sigma (V_t )}} \Biggr) \Biggr] -\newmathds{E} \bigl[Te^{-\int_{0}^{T}{r_tdt}}\varPhi
(X_{t_1},\ldots, X_{t_n} ) \bigr].\qedhere
\end{align*}
\end{proof}
%
It is also possible to compute the effect of other parameters on the
option by special choices of $\gamma$. For instance, as an immediate
result of Proposition~\ref{prp:Rho}, one can present the following two remarks.
\begin{remark}
Suppose that the assumptions in Proposition~\ref{prp:Rho} hold and the
stock price evolves from the Heston model with a stochastic interest
rate. Then, if $\gamma(X_t)=(0, \kappa, 0)^{\top}$, one obtains the
sensitivity of the option price with respect to $\kappa$.
\end{remark}
\begin{remark}
Suppose that the assumptions in Proposition~\ref{prp:Rho} hold and the
interest rate is assumed to follow the Vasicek interest rate model.
Then, if $\gamma(X_t)=(0, 0, a)^{\top}$, one obtains the effect of
``speed of reversion'' parameter on the option price.
\end{remark}
%
\subsection{Computation of Vega}
For the computation of Vega, one need to define a new perturbed process
as in the computation of Rho, but in this case, the perturbation will
occur in the diffusion term. However, in the end, it is necessary to
calculate the Skorohod integral. Hence, the result of Proposition~\ref
{prp:skorohod} will be used.

The perturbation approach in this section is based on the approach used
in~\cite{Davis2006,Fournie1999}. First, consider the perturbed asset
price process
\begin{equation}
\label{eq:varinvol} dX_t^\epsilon=\beta \bigl(X_t^\epsilon
\bigr)dt+ \bigl(a \bigl(X_t^\epsilon \bigr)+\epsilon\gamma
\bigl(X_t^\epsilon \bigr) \bigr)d\mathbb{W}_t,\qquad
X_0^\epsilon=x,
\end{equation}
where $\epsilon$ is a small scalar, $\gamma$ is a $3\times3$ matrix
valued continuously differentiable function with bounded derivatives.
Furthermore, $\beta$ and $(a+\epsilon\gamma)$ satisfy the
aforementioned regularity conditions. Here it is necessary to introduce
a variation process with respect to $\epsilon$, which is the derivative
of $X_t^\epsilon$ with respect to the parameter $\epsilon$,
$Z_t^\epsilon=\frac{\partial X_t^\epsilon}{\partial\epsilon}$,
\begin{align}
\label{vega:Zt} dZ_t^\epsilon&= \beta^\prime
\bigl(X_t^\epsilon \bigr)Z_t^\epsilon dt
\nonumber
\\
&\quad+\sum_{i=1}^{3}{ \bigl(a_i^\prime
\bigl(X_t^{\epsilon
} \bigr)+\epsilon\gamma_i^\prime
\bigl(X_t^{\epsilon} \bigr) \bigr)Z_t^\epsilon
dW_t^{i}}+\gamma \bigl(X_t^{\epsilon}
\bigr)d\mathbb{W}_t,
\\
\nonumber
Z_0^\epsilon&=\mathbh{0}_{3\times3}.
\end{align}
Here, $\gamma_i^\prime$ denotes the derivative of $i$th column.\vadjust{\goodbreak}

To avoid degeneracy,~the set $\tilde{\varGamma_n}$ of square integrable
functions in $\mathbb{R}$,
\begin{equation*}
\tilde{\varGamma_n}= \Biggl\{\tilde{\alpha} \in L^2
\bigl( [0, T ] \bigr):\int_{t_{i-1}}^{t_i}{\tilde{\alpha}
(t )dt}=1, \forall i=1,\ldots,n \Biggr\},
\end{equation*}
is defined in \cite{Fournie1999}. The following Proposition tells how
sensitive the price of an option on the perturbed process is to
$\epsilon$ in the point $\epsilon=0$.
\begin{proposition}
\label{prp:genvega}
Suppose that $a$ satisfies the uniform ellipticity condition (\ref
{eqn:ellipticity}) and for
$B_{t_i}=Y_{t_i}^{-1}Z_{t_i}=Y_{t_i}^{-1}Z_{t_i}^{\epsilon=0}$,
$i=1,\ldots,n$, there exists $a^{-1} (X )YB \in Dom(\delta)$.
Then, for any square integrable option payoff function, $\varPhi$, with
continuously differentiable and bounded derivatives,
\begin{equation*}
\frac{\partial}{\partial\epsilon}p^{\epsilon} (x )|_{\epsilon
=0}=\newmathds{E}
\bigl[e^{-\int_0^T{r_tdt}}\varPhi (X_{t_1},\ldots, X_{t_n} )\delta
\bigl(a^{-1} (X_. )Y_.\tilde{B}_. \bigr) \bigr]
\end{equation*}
holds. Here,
\begin{equation*}
\label{tildeB} \tilde{B}_t=\sum_{i=1}^{n}{
\tilde{\alpha} (t ) (B_{t_i}-B_{t_{i-1}} )\mathbh{1}_{ \{t\in[t_{i-1},t_i) \}}},
\end{equation*}
for $t_0=0$ and $\tilde{\alpha}\in\tilde{\varGamma_n}$. Moreover, if $B$
is Malliavin differentiable, the Skorohod integral is calculated
according to Remark~\ref{rmrk:trace} and it is
\begin{align*}
\delta \bigl(a^{-1} (X_. )Y_.\tilde{B}_. \bigr)&=\sum
_{i=1}^{n} \Biggl\{ B_{t_i}^{\top}
\int_{t_{i-1}}^{t_i}{\tilde{\alpha }(t)
\bigl(a^{-1} (X_t )Y_t \bigr)^{\top} d
\mathbb{W}_t}
\\
&\quad - \int_{t_{i-1}}^{t_i}{\tilde{\alpha}(t) Tr \bigl(
(D_t B_{t_i} )a^{-1} (X_t
)Y_t \bigr)dt}
\\
&\quad - \int_{t_{i-1}}^{t_i}{\tilde{\alpha}(t)
\bigl(a^{-1} (X_t )Y_t B_{t_{i-1}}
\bigr)^{\top} d\mathbb{W}_t} \Biggr\}.
\end{align*}
\end{proposition}
\begin{proof}
The proof can be found in~\cite{Davis2006}. 
\end{proof}
\begin{proposition}
Consider the three-dimensional SDE (\ref{eq:multiSDE}) and its
perturbed process (\ref{eq:varinvol}). Assume that $\beta$ and $a$ are
continuously differentiable functions with bounded derivatives and,
moreover, $a$ satisfies the uniform ellipticity condition (\ref
{eqn:ellipticity}). Then, Vega of an option with a square integrable
payoff function $\varPhi$, which is continuously differentiable with
bounded derivatives, is
\begin{align}
\mathit{Vega}^P&=\newmathds{E} \bigl[\varPhi (X_{t_1},\ldots,
X_{t_n} ) \mathit{Vega}_{\mathit{MW}} \bigr],
\end{align}
where
$\mathit{Vega}_{\mathit{MW}}$ is the Maliavin weight of Vega, and
\begin{align*}
\mathit{Vega}_{\mathit{MW}} &= e^{-\int_0^T{r_s ds}} \sum_{i=1}^{n}
\frac
{1}{t_i-t_{i-1}} \Biggl\{ \Biggl( \bigl(W_{t_i}^1 -
W_{t_{i-1}}^1\bigr) - \int_{t_{i-1}}^{t_i}
\sigma(V_t) dt \Biggr)
\\
&\quad\times \Biggl( \int_{t_{i-1}}^{t_i} \frac{1} {\sigma(V_t)}
dW_t^1 - \frac
{\rho_{12}}{\mu_1} \int_{t_{i-1}}^{t_i}
\frac{1}{ \sigma(V_t)} dW_t^2
\\
&\qquad+ \frac{ \rho_{12} \mu_2 - \rho_{13} \mu_1 }{\mu_1 \mu_3} \int_{t_{i-1}}^{t_i}
\frac{1}{ \sigma(V_t)} dW_t^3 \Biggr)- \int
_{t_{i-1}}^{t_i} {\frac{1} {\sigma(V_t)} dt } \Biggr\}.
\end{align*}
\end{proposition}
\begin{proof}
First obtain a perturbed process
by perturbing the original diffusion matrix with $\gamma$, where it is
chosen as
\begin{equation*}
\gamma (X_t )= %
\begin{bmatrix}
S_t & 0 & 0 \\[0.3em]
0 & 0 & 0 \\[0.3em]
0 & 0 & 0
\end{bmatrix} %
.
\end{equation*}
Then, using the fact that $V_t$ and $r_t$ do not depend on $S_t$, one
can deduce that the variation process $Z_t^{\epsilon}$ has vanishing
components; so in $Z_t^{\epsilon=0}$ the components $Z_t^2$ and $Z_t^3$
are almost surely zero. Then, as in Proposition~\ref{prp:genvega},
define the vector
$B_{t_i}=Y_{t_i}^{-1}Z_{t_i}=Y_{t_i}^{-1}Z_{t_i}^{\epsilon=0}$,
$i=1,\ldots,n$ for $t_{i} \in[0,T]$. Here,
\begin{equation*}
Y_{t_i}^{-1}= %
\begin{bmatrix}
\frac{1}{Y_{t_i}^{11}} & -\frac
{Y_{t_i}^{12}}{Y_{t_i}^{11}Y_{t_i}^{22}} & -\frac
{Y_{t_i}^{13}}{Y_{t_i}^{11}Y_{t_i}^{33}} \\[0.3em]
0 & \frac{1}{Y_{t_i}^{22}}& 0 \\[0.3em]
0 & 0 & \frac{1}{Y_{t_i}^{33}}
\end{bmatrix} %
,
\end{equation*}
and
\begin{equation*}
Z_{t_{i}}= %
\begin{bmatrix}
Z_{t_{i}}^1 \\[0.3em]
0 \\[0.3em]
0
\end{bmatrix} %
.
\end{equation*}
Then, substituting these two equation into $B_{t_i}$, one obtains
\begin{equation}
\label{betaVega} B_{t_i}= %
\begin{bmatrix}
\frac{Z_{t_i}^1}{Y_{t_i}^{11}} \\[0.3em]
0 \\[0.3em]
0
\end{bmatrix} %
.
\end{equation}
On the other hand, from equation (\ref{vega:Zt}) it is known that
$Z_{t_i}$ satisfies the following dynamics
\begin{align*}
\begin{bmatrix}
dZ_t^1 \\[0.3em]
dZ_t^2 \\[0.3em]
dZ_t^3
\end{bmatrix} %
&= %
\begin{bmatrix}
r_t & 0 & S_t \\[0.3em]
0& u^{\prime} (V_t )& 0 \\[0.3em]
0& 0 & f^{\prime} (r_t )
\end{bmatrix}
\begin{bmatrix}
Z_t^1 \\[0.3em]
0\\[0.3em]
0
\end{bmatrix} %
dt
\\
&\quad+ %
\begin{bmatrix}
\sigma (V_t ) & S_t\sigma^{\prime} (V_t )& 0\\[0.3em]
0 & \rho_{12} v^{\prime} (V_t ) & 0 \\[0.3em]
0 & 0 & \rho_{13}g^{\prime} (r_t )
\end{bmatrix} %
\begin{bmatrix}
Z_t^1 \\[0.3em]
0\\[0.3em]
0
\end{bmatrix}
dW_t^1
\\
&\quad+ %
\begin{bmatrix}
0 & 0 & 0\\[0.3em]
0 & \mu_{1} v^{\prime} (V_t ) & 0 \\[0.3em]
0 & 0 & \mu_2 g^{\prime} (r_t )
\end{bmatrix} %
\begin{bmatrix}
Z_t^1 \\[0.3em]
0\\[0.3em]
0
\end{bmatrix}
dW_t^2\\
&\quad + %
\begin{bmatrix}
0 & 0& 0\\[0.3em]
0 & 0 & 0 \\[0.3em]
0 & 0 & \mu_3 g^{\prime} (r_t )
\end{bmatrix}
\begin{bmatrix}
Z_t^1 \\[0.3em]
0\\[0.3em]
0
\end{bmatrix} %
dW_t^3
+ %
\begin{bmatrix}
S_t & 0& 0\\[0.3em]
0 & 0 & 0 \\[0.3em]
0 & 0 & 0
\end{bmatrix} %
d\newmathds{W}_t,
\end{align*}
for $ t \in[0,T]$. From this setting, one can write
\begin{equation*}
dZ_{t_i}^1=r_{t_i} Z_{t_i}^1
dt+ \sigma (V_{t_i} )Z_{t_i}^1 dW_{t_i}^1+S_{t_i}
dW_{t_i}^1.
\end{equation*}
With the It\^{o} formula, one can easily find the solution as
\begin{equation*}
Z_{t_i}^{1}=S_{t_i} \Biggl(W_{t_i}^1-
\int_0^{t_i}{\sigma (V_s )ds} \Biggr).
\end{equation*}
Thus, using equation (\ref{betaVega}) and Remark \ref{rmrk:delta}, one has
\begin{equation*}
B_{t_i}^1=\frac{ S_{t_i} (W_{t_i}^1-\int_0^{t_i}{\sigma
(V_s )ds} )} {Y_{t_i}^{11}}= S_0
\Biggl(W_{t_i}^1-\int_0^{{t_i}}{
\sigma (V_s )ds} \Biggr),
\end{equation*}
for $t_i \in[0,T]$. According to Proposition~\ref{prp:genvega}, the
Skorohod integral $\delta(a^{-1} (X)Y\tilde{B})$ remains to be
calculated. Here, $B_{\cdot}$ is Malliavin differentiable and its
Malliavin derivative is
\begin{align*}
D_t B_{t_i}^1&= S_0 \Biggl(
(1,0,0 )-\int_{0}^{t_i}{D_t \sigma
(V_s )ds} \Biggr),
\\
&= S_0 \Biggl( (1,0,0 )- \int_{0}^{t_i}{
\sigma^{\prime} (V_s ) D_t{V_s} \, ds}
\Biggr)
\\
&= S_0 \Biggl( (1,0,0 )- \int_{0}^{t_i}{
\sigma^{\prime} (V_s ) \biggl(\rho_{12}
v(V_t) \frac{Y_s^{22}}{Y_t^{22}}, \mu_1 v(V_t)
\frac{Y_{s}^{22}}{{Y_t}^{22}}, 0\biggr) ds} \Biggr).
\end{align*}
Then, one obtains the trace
\begin{equation*}
Tr \bigl( (D_t B_{t_i} )a^{-1}(X_t)Y_{t}
\bigr)=\frac{1}{\sigma
 (V_{t} )}.
\end{equation*}
As a result, by choosing $\tilde{\alpha}= \frac{1}{t_i-t_{i-1}}$,
\begin{align*}
\delta\bigl(a^{-1} (X) Y \tilde{B}\bigr)&= \sum
_{i=1}^{n} \frac{1}{t_i-t_{i-1}} \Biggl\{ \Biggl(
\bigl(W_{t_i}^1 - W_{t_{i-1}}^1\bigr) - \int
_{t_{i-1}}^{t_i} \sigma (V_t) dt \Biggr)
\\
&\quad\times \Biggl( \int_{t_{i-1}}^{t_i} \frac{1} {\sigma(V_t)}
dW_t^1 - \frac
{\rho_{12}}{\mu_1} \int_{t_{i-1}}^{t_i}
\frac{1}{ \sigma(V_t)} dW_t^2
\\
&\qquad+ \frac{ \rho_{12} \mu_2 -\rho_{13} \mu_1 }{\mu_1 \mu_3} \int_{t_{i-1}}^{t_i}
\frac{1}{ \sigma(V_t)} dW_t^3 \Biggr) - \int_{t_{i-1}}^{t_i} {\frac{1} {\sigma(V_t)} dt } \Biggr
\}.\qedhere
\end{align*}
\end{proof}
\begin{remark}
If the option payoff depends on only the maturity $T$, then
\begin{align*}
\mathit{Vega}^P&=\newmathds{E} \bigl[\varPhi (X_{T} )
\mathit{Vega}_{\mathit{MW}} \bigr],
\end{align*}
where
\begin{align*}
\mathit{Vega}_{\mathit{MW}} &= \frac{e^{-\int_0^T{r_s ds}}} {T} \Biggl\{ \Biggl(W_T^1-
\int_{0}^{T}{\sigma (V_t )dt} \Biggr)
\\
&\quad\times \Biggl( \int_{0}^{T} \frac{1} {\sigma(V_t)}
dW_t^1 - \frac{\rho_{12}}{\mu_1} \int_{0}^{T}
\frac{1}{ \sigma(V_t)} dW_t^2
\\
&\qquad + \frac{ \rho_{12} \mu_2 - \rho_{13} \mu_1 }{\mu_1 \mu_3} \int_{0}^{T}
\frac{1}{ \sigma(V_t)} dW_t^3 \Biggr) - \int
_{0}^{T}{\frac
{1}{\sigma (V_t )}dt} \Biggr\}.
\end{align*}
\end{remark}

\section{Numerical illustration}
\label{sec:num}
This section is devoted to numerical illustrations of the Greeks of a
European call option with a strike price $K$ and an option payoff
function $\varPhi=\max \{S_T-K, 0 \}$, where $S_T$ is the price
of the underlying asset at maturity $T<\infty$.

It is also worth to emphasize that a European call option has a
Lipschitz payoff function and it belongs to the space of locally
integrable functions denoted by $L^2$. Further, the space $C_c^\infty$
of infinitely differentiable functions having a bounded compact
support, where $c$ is an arbitrary compact subset of $\mathbb{R}^3$, is
dense in $L^2$. Hence, there exists a sequence of functions $\varPhi_n
\subset C_c^\infty$ that converges to the main payoff function $\varPhi
\in L^2$,~see~\ref{prp:genchain}. Therefore, one can apply the formulas
introduced in the previous section to European options.

The formulas introduced in Section~\ref{sec:MCgreeks} are for a general
case since the functions in the SDEs (\ref{eq:general1})--(\ref
{eq:general3}) are given with closed forms. Therefore, it is possible
to find the Greeks for all stochastic volatility models through the
special choice of functions that are introduced in SDEs (\ref
{eq:general1})--(\ref{eq:general3}). In this study, these functions are
chosen according to the well-known Heston stochastic volatility model
with a stochastic interest rate, namely the Vasicek model for the
simulation purposes. Under these special choices, the SDEs (\ref
{eq:general1})--(\ref{eq:general3}) become
\begin{align}
dS_t&=r_{t}S_{t}+S_{t}
\sqrt{V_t}dW_{t}^{1},
\\
dV_t&=\kappa (\theta-V_t )dt+\sigma
\sqrt{V_t} \bigl(\rho _{12}\;dW_{t}^{1}+
\mu_1\;dW_{t}^{2} \bigr),
\\
dr_t&=a (b-r_t )dt+k \bigl(\rho_{13}dW_t^1+
\mu_2dW_t^2+\mu _3dW_{t}^{3}
\bigr),
\end{align}
where the initial values are $S_0, V_0$ and $r_0$, respectively. Here,
it is assumed that the coefficients $\kappa, \theta,\sigma,a,b$, and
$k$ are all positive numbers.

Under this setting, the functions $\beta$ and $a$ in (\ref
{eq:multiSDE}) are continuously differentiable, and satisfy the
Lipschitz condition. Moreover, $a$ satisfies the uniform ellipticity
condition (\ref{eqn:ellipticity}). These assumptions are enough to
compute the Greeks in the BSM framework. However, one needs more
assumptions in the Heston stochastic volatility model framework because
the square root function is not differentiable and not globally
Lipschitz. The Novikov condition for the Heston stochastic volatility
model, $\kappa\theta\geq\sigma^2$, guarantees that the volatility
process is always positive. Hence, it is assumed that the Novikov
condition is satisfied, and the initial volatility $V_0$ is positive.
Moreover, in the studies of~\cite{Alos2008,ewald2009malliavin}, it is
proved that the Heston stochastic volatility model is Malliavin
differentiable under the Novikov condition. In the rest of the paper,
it is assumed that the Heston stochastic volatility model satisfies the
Novikov condition and it is Malliavin differentiable.

\subsection{Figures}

In the numerical applications, the model parameters are set as follows:
$\rho_{12}=-0.8$, $\rho_{13} = 0.5$, $\rho_{23}=0.02$, $K=100$,
$S_0=100$, $V_0=0.04$, $R_0=0.02$, $\kappa=2$, $\theta=0.04$, $b=0.08$,
$a=0.02$ $k=0.002$ $r=0.05$ $\sigma=0.04$ $T=1$. The number of
simulations, i.e. the number of paths, is set as 10000, although
convergence of our approach is already pretty good even if it is set as
only 250. Finally, the numbers of discretization steps are set as 252
considering the trading days in a year.

Using the above given values, the study presents simulations for the
Greeks Delta, Rho and Vega of a European call option after conducting
successfully the computations in the previous section. Figures~\ref
{fig:MdeltaBlack},~\ref{fig:MRhoBlack} and~\ref{fig:MVegaBlack} allow
comparing the finite difference method in all variations and the
Malliavin calculus on sample sizes. The computed Delta, Rho and Vega
values of both methods are very stable and quite good, even for a low
number of MC simulations. Furthermore, if the number of simulation
increases, the Greeks values become more stable for each method.
Therefore, one should increase the number of simulations for both
methods to have a more accurate value for the Greeks.

\begin{figure}[t!]
\includegraphics{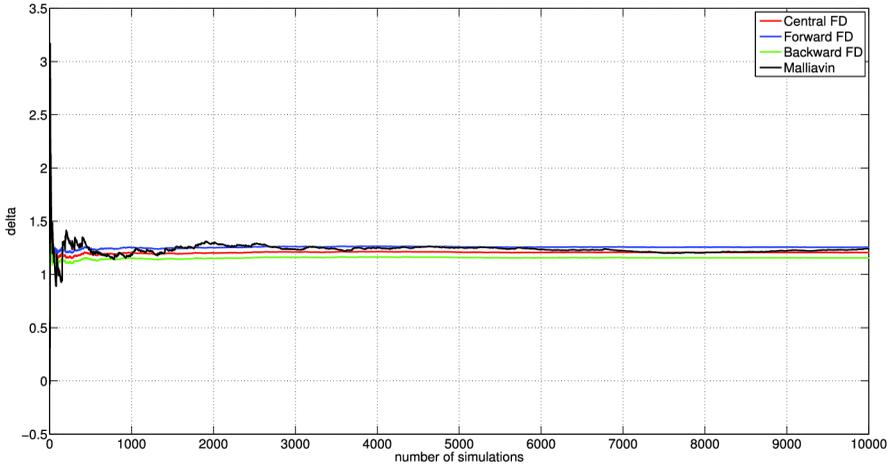}
\caption{Delta of a European call option in the SHV model with the
Malliavin calculus and the finite difference method in all variations}
\label{fig:MdeltaBlack}
\end{figure}

\begin{figure}[t!]
\includegraphics{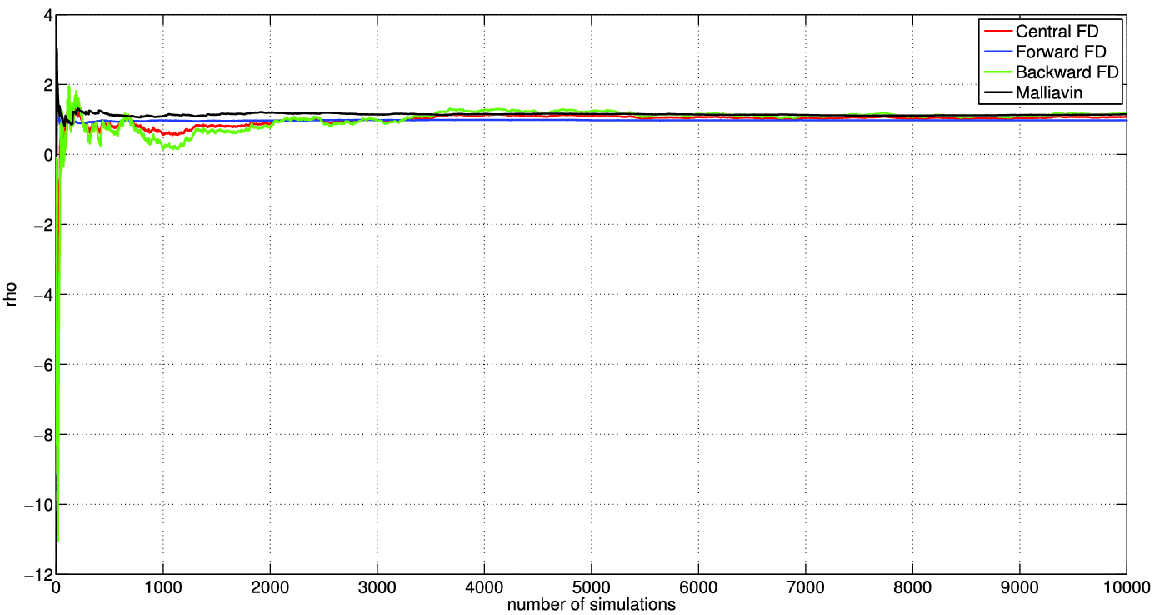}
\caption{Rho of a European call option in the SHV model with the
Malliavin calculus and the finite difference method in all variations}
\label{fig:MRhoBlack}
\end{figure}

\begin{figure}[t!]
\includegraphics{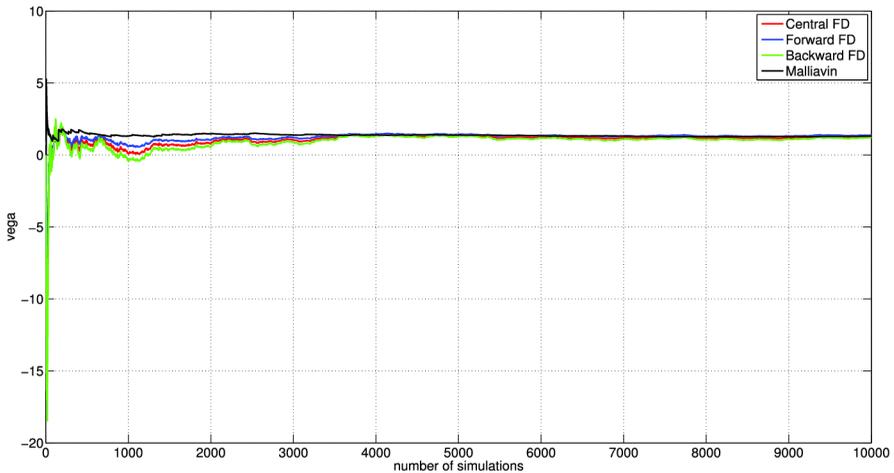}
\caption{Vega of a European call option in the SHV model with the
Malliavin calculus and the finite difference method in all variations}
\label{fig:MVegaBlack}
\end{figure}

\section{Conclusion}
\label{sec:conclusion}
In this study, a very general model skeleton is considered in equations
(\ref{eq:general1}), (\ref{eq:general2}), and (\ref{eq:general3}) for
the computation purposes, and the formulas with detailed proofs for the
well-known Greeks under the assumptions of stochastic hybrid volatility
models are derived. As in many other studies, the Greeks are obtained
as an expectation of product of two terms: a payoff function and a
weight called the Malliavin weight which is independent of the payoff
function. This result indicates that the efficiency of the Malliavin
calculus in the computation of the Greeks does not depend on the type
of the payoff function. In the computations, it is assumed that the
payoff function is continuously differentiable. However, in the case of
mathematical finance, payoff functions are not globally differentiable,
so it is better to mention that explicit expressions given in this
paper are easily extended to payoff functions which are not
continuously differentiable. Hence, once the formulas are obtained for
one option, one can use the same formula for all kinds of options by
changing only the payoff function. Moreover, by substituting the
necessary functions into the Greeks formulas, one can obtain the Greeks
for all kinds of models, and these formulas can be easily adapted to
the special needs of financial engineers working in practice on the
computation of Greeks. As an application of the results, the paper also
examines a particular case of the Heston stochastic volatility model by
assuming the interest rate evolving from the Vasicek model. In order to
compare the results, these Greeks are computed with the finite
difference method in all variations. It is observed that the formulas
which are obtained by using the Malliavin calculus yield results that
require a fewer number of simulations than in the finite difference
method for Vega and Rho. Moreover, despite the easy implementation of
the finite difference method, the duration of computation is higher
than in the Malliavin calculus. Since traders need the Greeks for
hedging purposes, they have to be computed as fast as possible. It is,
therefore, using the Mallaivin calculus is superior to the finite
difference method in the computation of Greeks because once the
formulas are obtained, they can be used for all types of option and the
duration of computation is shorter than in the finite difference method.

\begin{appendix}
\section{A brief review on Malliavin calculus}
\label{sec:preliminaries}

\begin{definition}
\label{def:deriv}
Let $ F=f (W(h_1),\ldots,W(h_n) ) \in\mathcal{S}$ with
$H=L^2([0,T], \mathcal{B}, \mu)$. Then the derivative $D:\mathcal
{S}\mapsto L^2 (\varOmega\times [0, T ] )$ of $F$ is
defined by
\begin{equation*}
DF=\sum_{i=1}^{n}{\frac{\partial}{\partial x_i}f
\bigl(W(h_1),\ldots, W(h_n) \bigr)h_i},
\end{equation*}
where $\frac{\partial f} {\partial x_i}$ is the partial derivative of
$f$ with respect to its $i$th variable.
\end{definition}
\begin{proposition}[Chain Rule](Proposition 1.2.3 in~\cite
{nualart2006malliavin})
\label{prp:chain}
Suppose that $F\,{=}\,(F_1,\ldots,\break F_n)$ is a random vector whose components
belong to the closure of $\mathcal{S}$, $\mathbb{D}^{1,2}$, and the
function $\varphi:R^n\mapsto R$ is a continuously differentiable
function with bounded partial derivatives. Then $\varphi(F)\in\mathbb
{D}^{1,2}$and
\begin{equation*}
D_t \varphi(F)=\sum_{i=1}^{n}{
\frac{\partial}{\partial x_i}\varphi(F)\; D_tF_i}= \bigl\langle\nabla
\varphi(F),DF \bigr\rangle,
\end{equation*}
almost surely for $t \in[0,T]$.
\end{proposition}
%
%
\begin{proof}
If the function $\varphi$ is smooth the proof can be obtained by the
chain rule in classical analysis. Otherwise, the function has to be
mollified. In order to mollify $\varphi$, one can use $\rho_\epsilon
(x)=\epsilon^n\rho(\epsilon x)$, where $\rho(x)=ce^{\frac{1}{x^2-1}}$
and $c$ is a chosen coefficient that makes the integral $\int_{\mathbb
{R}^n}\rho(x)dx=1$, to obtain a smooth approximation $\varphi*\rho
_\epsilon$. Considering the smooth approximations $F_n$ of $F$ one
obtains $\varphi*\rho_\epsilon(F_n)\mapsto\varphi(F)$ for $\min{\epsilon
, n}\mapsto\infty$ in the space $L^2$. Then by closedness of the
derivative operator $D$\vadjust{\goodbreak}
\begin{align*}
& \Biggl\llVert D\varphi(F)-\sum_{i=1}^{n}{
\frac{\partial}{\partial x_i}\varphi (F)DF_{ni}} \Biggr\rrVert _2
\\
&\quad\leq \bigl\llVert D\varphi(F)-D\varphi*\rho_\epsilon(F) \bigr\rrVert
_2
\\
&\qquad+ \Biggl\llVert D\varphi*\rho_\epsilon(F)-\sum
_{i=1}^{n}{\frac{\partial
}{\partial x_i}\varphi*
\rho_\epsilon(F_n)DF_{ni}} \Biggr\rrVert
_2
\\
&\qquad+\Biggl\llVert \sum_{i=1}^{n}{
\frac{\partial}{\partial x_i}\varphi*\rho _\epsilon}(F_n)DF_{ni}-
\sum_{i_1}^{n}{\frac{\partial}{\partial}\varphi
(F)DF_i} \Biggr\rrVert _2\mapsto0.\qedhere
\end{align*}
\end{proof}
%
%
\begin{lemma}
\label{lemm:cnv}
Suppose that the sequence $F_n\in\mathbb{D}^{1,2}$ converges to $F$ in
the space $L^2 (\varOmega,\mathcal{F},\mathbb{P} )$ satisfying
$\sup\limits_{n}\newmathds{E} [ \llVert DF \rrVert _H^2 ]<\infty$.
Then, $F\in\mathbb{D}^{1,2}$ and $DF_n$ weakly converges to $DF$ in
$L^2 (\varOmega\times [0, T ] )$.
\end{lemma}

A generalization of Proposition~\ref{prp:chain} to the functions even
not necessarily differentiable is given below.
%
\begin{proposition}(Proposition 1.2.4 in~\cite{nualart2006malliavin})
\label{prp:genchain}
Given a function $\varphi$ that satisfies, for a positive constant $K\in
\mathbb{R}$,
\begin{equation*}
\bigl\llvert \varphi(x)-\varphi(y) \bigr\rrvert \leq K \llvert x-y \rrvert ,
\quad x,y\in \mathbb{R}^n
\end{equation*}
and $F\in\mathbb{D}^{1,2}$. Then $\varphi(F)\in\mathbb{D}^{1,2}$ and
there exists an n-dimensional random vector $G\in\mathbb{R}^n$, $ \llvert G \rrvert <K$ such that
\begin{equation*}
D\bigl(\varphi(F)\bigr)=\sum_{i=1}^{n}{G_iDF_i}.
\end{equation*}
\end{proposition}
\begin{proof}
Using the same mollifier $\rho_\epsilon$ as defined in the proof of
Proposition~\ref{prp:chain}, one can obtain $\varphi*\rho_\epsilon$
that converges to $\varphi$. The sequence $D(\varphi*\rho_\epsilon)(F)$
is bounded in the space $L^2 (\varOmega\times [0, T ]
)$. This is because $ \llvert \nabla(\varphi*\rho_\epsilon) \rrvert \leq K$
for some large $\epsilon$. From Lemma~\ref{lemm:cnv}, $\varphi(F)\in
\mathbb{D}^{1,2}$ and the Malliavin derivative $D(\varphi*\rho_\epsilon
)(F)\mapsto D(\varphi(F))$ in the weak sense. On the other hand, $\nabla
(\varphi*\rho_\epsilon)(F)$ converges weakly to a vector $G\in\mathbb
{R}^n$, $ \llvert G \rrvert <K$. So, one can take the the weak limit in
\begin{equation}
D(\varphi*\rho_\epsilon) (F)=\sum_{i=1}^n{
\frac{\partial}{\partial
x_i}\varphi(F)DF_i}
\end{equation}
to lead us to the result.
\end{proof}

\begin{definition}[Skorohod Integral]
Consider $u\in L^2 ([0,T ]\times\varOmega )$. Then,
$u\in Dom(\delta)$ if for all $F\in\mathbb{D}^{1,2}$ and
\begin{equation*}
\Biggl\llvert \newmathds{E} \Biggl[\int_{0}^{T}{D_tF
u_t dt} \Biggr] \Biggr\rrvert \leq c \llVert F \rrVert
_{L^2(\varOmega)},
\end{equation*}
where $c$ is some constant depending on $u$,
\begin{equation*}
\delta(u)=\int_0^T{u_t\delta
W_t}
\end{equation*}
is an element of $L^2(\varOmega)$, and the duality formula holds:
\begin{equation*}
\newmathds{E} \Biggl[\int_0^T{D_tF
u_t dt} \Biggr]=\newmathds{E} \bigl[F\delta (u) \bigr],\quad\forall F\in
\mathbb{D}^{1,2}.
\end{equation*}
\end{definition}
%
%
\begin{proposition}[Integration by Parts Formula]
\label{prp:skorohod}
Suppose $F\in\mathbb{D}^{1,2}$, $F h \in Dom(\delta) $ for $h\in H$. Then
\begin{equation*}
\delta (Fh )=FW(h)- \langle DF, h \rangle_H.
\end{equation*}
Moreover if $F=1$ a.s.,
\begin{equation*}
\delta(h)=W(h).
\end{equation*}
\end{proposition}
%
%
\begin{remark}
Note that, in particular, if $H=L^2([0,T], \mathcal{B}, \mu)$, where
$\mu$ is a $\sigma$-finite atomless measure on a measurable Borel space
$([0,T], \mathcal{B})$, then Proposition \ref{prp:skorohod} turns to
\begin{equation*}
\int_{0}^{T} F h_t \delta
W_t =F \int_{0}^{T} h_t
\delta W_t -\int_{0}^{T}{D_tF
h_t dt}.
\end{equation*}
\end{remark}
%
%
\begin{remark}
\label{rmrk:delta2}
The domain of the Skorohod integral also contains the adapted
stochastic processes in $L^2([0,T ]\times\varOmega)$. When the
integrand is adapted, then the Skorohod integral coincides with the
It\^{o} integral (see \cite{nualart2006malliavin}), i.e.
\begin{equation*}
\delta(h)=\int_{0}^{T}{h_t
dW_t},
\end{equation*}
and
\[
\int_{0}^{T} F h_t dW_t =
F\int_{0}^{T}{h_t dW_t}-
\int_{0}^{T}{D_tF h_t dt},
\]
for $F\in\mathbb{D}^{1, 2}$ and $\newmathds{E} [\int_0^T{(F h_t
)^2dt} ]<\infty$.
\end{remark}
%
%
\begin{proposition}(Proposition 1.3.8 in~\cite{nualart2006malliavin})
Let $u\in\mathbb{L}^{1,2}=\mathbb{D}^{1,2}(L^2(T))$ be a stochastic
process satisfying $\newmathds{E} [\int_0^T{u^2(s,\omega)ds}
]<\infty$ for all $0\leq s\leq T$. Assume that $D_tu\in Dom(\delta)$
for all $0\leq t\leq T$
and $\newmathds{E} [\int_0^T{ (\delta (D_tu )
)^2dt} ]<\infty$. Then, $\delta(u)\in\mathbb{D}^{1, 2}$ and
\begin{equation*}
D_t \bigl(\delta(u) \bigr)=u(t, \omega)+\int_{0}^{T}{D_tu(s,
\omega)dW_s}.
\end{equation*}
Moreover, if $u(t, \omega)$ is an adapted process belonging to $L^2([0,T]\times\varOmega)$, then\querymark{Q2}
\begin{equation*}
D_t \Biggl(\int_{0}^{T}{u(s,
\omega)dW_s} \Biggr)=u(t,\omega)+\int_{t}^{T}{D_tu(s,
\omega)dW_s}.
\end{equation*}
\end{proposition}
%
%
\begin{remark}
\label{rmrk:trace}
Suppose $F$ is a d-dimensional random column-vector and $u_t$ is a
$d\times d$ matrix-process. Then Remark~\ref{rmrk:delta} translates to
\begin{equation*}
\delta(Fh)= F*\int_0^T{h_td
W_t}-\int_0^T{Tr(D_tF)h_t}dt,
\end{equation*}
with the convention that the It\^o integral for a matrix process is a
column-vector~\cite{Davis2006}.
\end{remark}
\end{appendix}












\end{document}